\documentclass[pra,reprint,onecolumn,floatfix,notitlepage,showkeys]{revtex4-1}

\usepackage{graphicx}
\usepackage{amssymb}
\usepackage{epsfig}
\usepackage{epsf}
\usepackage{epstopdf}
\usepackage{xcolor}

\newtheorem{theorem}{Theorem}[section]
\newtheorem{proof}{Proof}[section]

\newtheorem{lemma}{Lemma}[section]

\begin{document}
\title{Modified Covariance Intersection for Data Fusion in Distributed Non-homogeneous Monitoring Systems Network}

\author{Abolghasem Daeichian}
\email{ E-mail:a-daeichian@araku.ac.ir, a.daeichian@gmail.com}
\thanks{This is the peer reviewed version of the following article: Daeichian, Abolghasem, and Elham Honarvar. "Modified covariance intersection for data fusion in distributed nonhomogeneous monitoring systems network." International Journal of Robust and Nonlinear Control 28.4 (2018): 1413-1424., which has been published in final form at DOI: 10.1002/rnc.3964. This article may be used for non-commercial purposes in accordance with Wiley Terms and Conditions for Use of Self-Archived Versions.}
\affiliation{Department of Electrical Engineering, Faculty of Engineering, Arak University, Arak, 38156-8-8349 Iran}
\author{Elham Honarvar}
\affiliation{Payam Nonprofit Higher Education
	Institute, Golpayegan, Iran}

\begin{abstract}
Monitoring networks contain monitoring nodes which observe an area of interest to detect any possible existing object and estimate its states. Each node has characteristics such as probability of detection and clutter density which may have different values for distinct nodes in non-homogeneous monitoring networks. This paper proposes a modified covariance intersection method for data fusion in such networks. It is derived by formulating a mixed game model between neighbor monitoring nodes as players and considering inverse of the trace of fused covariance matrix as players' utility function. Monitoring nodes estimate the states of any possible existing object by applying joint target detection and tracking filter on their own observations. Processing nodes fuse the estimated states received from neighbor monitoring nodes by the proposed modified covariance intersection. It is validated by simulating target detection and tracking problem in two situations: one-target and unknown number of targets.
\end{abstract}
\maketitle
\section{Introduction}
Precision promotion and cost reduction for detection and tracking of moving objects (targets) in a noisy environment is a challenging problem. Nowadays, distributed monitoring networks are extensively used in target detection and tracking in both civilian and military applications due to providing scalability, flexibility, low operation cost, and diversity in viewing geometry and phenomenology \cite{dong2014distributed,bocca2014multiple}. Monitoring networks consist of monitoring nodes, processing nodes (fusion center), and communication links. Monitoring nodes contain sensors which constantly measure some features in the interested space. The measurements are used to detect and track any possible existing object which results in an estimation of its states. Processing nodes fuse estimations of two or more monitoring nodes which are received through communication links in order to give more precise information.
Monitoring nodes may connect to processing nodes in different topologies; for example, centralized topology which is theoretically optimal but has some disadvantages such as high bandwidth reqirement for collecting all measurements in a single node, high computation load at a single location, high power consumption, and low robustness due to a single point of failure. Decentralized or distributed data fusion is an alternative topology which consists of multi local processing nodes \cite{ding2014recent}.

Object tracking consists of detecting and then tracking an object \cite{ulker2012multiple,morelande2007bayesian}. A unified target detection and tracking has been introduced as joint target detection and tracking (JoTT) filter to track a target which has Gaussian \cite{bazzazzadeh2009optimal} or Bernoulli \cite{mahler2007statistical,vo2014labeled} characteristics.
%JoTT filter decides whether target is presented in observation volume or not and if there is, extracts its states from the observations \cite{morelande2007bayesian}.
Exact computing of probability density is a drawback of this method. Monte-Carlo or particle filter techniques tackle this disadvantage due to their capability of approximating a probability density efficiently by a cloud of $N$ weighted particles \cite{olsson2011rao,han2013particle,li2015multiple}. Rao-Blackwellized particle filtering (RBPF) is an extension on particle filter that not only gives data associations and estimation of target states, but also model births and deaths of the targets as hidden stochastic processes observed through the measurements \cite{sarkka2007rao}.
In the situation of unknown number of targets, several factors such as number of targets, corrupted measurements by clutter, and target appearance-disappearance from time to time must be considered \cite{ulker2012multiple}. A difficulty of unknown-target joint detection and tracking is due to the fact that number of targets and measurements both vary randomly in time. Thus, it is not clear that each measurement is generated by which target \cite{vo2004joint}. Some researches such as \cite{oh2004markov,oh2009markov} cope with this conundrum.

Distributed data fusion (DDF) and filtering algorithms are the key components of any target tracking system \cite{durrant2005data}. Data fusion techniques merge data from multiple sensors and related information to provide more accurate and applicable data in comparison with using a single, independent sensor \cite{waltz1990multisensor,khaleghi2013multisensor}. This allows either improved accuracy from existing sensors or the same performance from smaller or cheaper sensors. Improving the observability and developing the observation space are another advantages of multi distributed sensor data fusion \cite{ahmed2013bayesian}. Unknown degree of correlation between estimates obtained from different monitoring nodes and double counting of the information that is previously used (redundant information) within an ad-hoc network topology are some challenges in practical implementation of DDF. Covariance Intersection (CI) presents suboptimal data fusion algorithm which addresses these challenges and avoids the assumption of independence of estimates required by traditional Bayesian filters Data fusion \cite{chen2002estimation,liggins2008handbook}. CI conducts fusion process by weighting the estimates of any monitoring node using a mixing parameter $\omega$. Usually, mixing parameter is found such that minimizing the trace (or determinant) of the fused covariance matrix which is equivalent to minimizing the Shannon entropy of the fused covariance matrix \cite{hurley2002information}. An iterative extension of the CI for DDF is given in \cite{hlinka2014distributed} which converges asymptotically to a consensus across all network nodes. A batch CI scheme is presented in \cite{sun2016distributed} to handle the unknown cross-correlation in DDF by means of an average consensus algorithm.

In many applications, monitoring nodes have different characteristics, capabilities, or with different levels of trust, i.e. non-homogeneous nodes. For instance, Network Centric Warfare (NCW) requires a sensor network which collects and fuses vast amount of disparate and complementary data from different non-homogeneous sensors that are geographically dispersed throughout the battle space \cite{julier2006challenge,smith2006approaches}. Some investigations try to deal with this challenge and keep away from the assumption of homogeneity of sensors. For example, a method for optimizing the mixing parameter to minimize information lost during fusion in a network with different measurement covariance matrix for each sensor is presented in \cite{clarke2016minimum}. In particular, optimality of fusion rules implemented in detection and tracking systems usually relies on probability of detection, clutter density and knowledge of probability distributions for all distributed sensors \cite{aziz2014new,smith2006approaches}.

This paper proposes a modified CI for data fusion in a non-homogeneous distributed monitoring network. Difference between detection parameters of distinct nodes such as probability of detection take into account in data fusion algorithm. To this end, a mixed game is articulated between two distinct monitoring nodes $a$ and $b$ as players and inverse of trace of fused covariance matrix as their utility function. Matching the utility function of the game with traditional CI leads to a modified CI which the accuracy of its fused estimates outperforms the traditional CI in the sense of fused covariance matrix. The modifies CI is also extended for more than two nodes. Detection and tracking in two circumstances (one target and unknown number of targets) in a cluttered environment are investigated to validate the proposed method. An interested space is observed by non-homogeneous monitoring nodes. The RBPF is employed to estimate target states in each node. The estimated states are transmitted to neighbor processing nodes via communication links. The received estimates are fused by the proposed modified-CI in processing nodes. The simulation results are compared with Kalman filter, particle filter, and traditional CI to verify utility of the proposed method.

The paper is folded as follows. Problem is stated in section \ref{section-ProblemStatement}. RBPF and RBMCDA algorithms for target detection and tracking are reviewed in section \ref{section-RBPF}. Section \ref{section-modCI} is devoted to representation of the proposed method. Simulation results are given in section \ref{section-Simulation}. Finally, the paper is concluded in section \ref{section-Conclusion}.

\section{Problem  statement}
\label{section-ProblemStatement}
Let a network of distributed and non-homogeneous monitoring nodes observe an area of interest to find targets with states $x_k\in \mathbb{R}^n$ at time $t_k$; e.g. see Fig.\ref{Network}. Each monitoring node has its own probability of detection, density of clutter, and $n_y$ measurements of the target state $x_k$, in vector $y_k\in \mathbb{R}^{n_y}$ at time $t_k$. They estimate mean ($m_k$) and covariance ($P_k$) of the states of any observed object using joint target detection and tracking filter. 

The estimates are sent to neighbor processing nodes through communication links. Processing nodes fuse received data to yield more accurate estimation. Different values for detection parameters such as probability of detection for each monitoring node may corrupt the outcome of fusion process. This paper concerns on how data fusion could be performed by considering non-homogeneity of monitoring node parameters, especially in detection parameters.

\begin{figure}[htb]
\begin{center}
\includegraphics[width=0.4\textwidth]{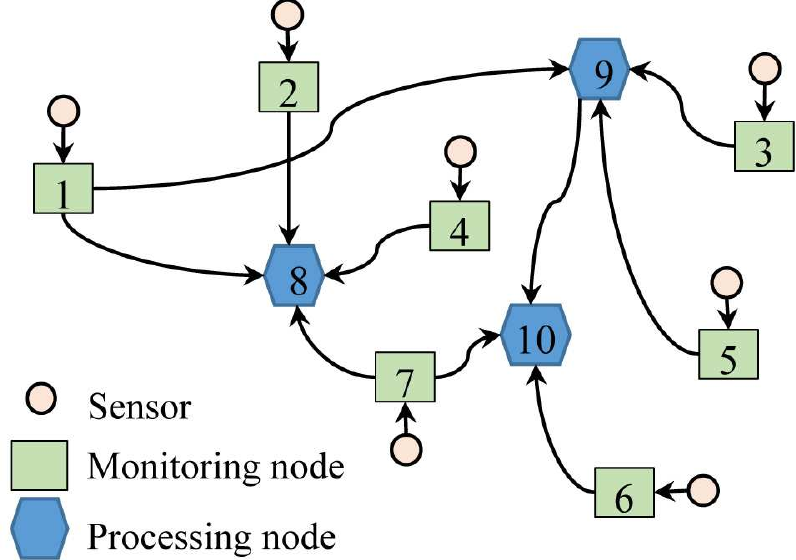}
\caption{Network of monitoring and processing nodes}
\label{Network}
\end{center}
\end{figure}

Consider the dynamic of target states as:
\begin{eqnarray}\label{EqSystem}
x_{k} &=& A_{k-1}x_{k-1}+ q_{k-1}
\end{eqnarray}
and sensor measurement model as:
\begin{eqnarray}\label{EqMeasurement}
y_k &=& H_k x_k+r_k
\end{eqnarray}
where $q_{k-1}$ and $r_k$ are states and measurements noise, respectively. This kind of noise suppose to be finite second moment stochastic processes with covariance matrices $Q_{k-1}$ and $R_k$, respectively. Let $c_k$ be data association indicator where $c_k=0$ represents the clutter and $c_k=j$ indicates target number $j$.
Clutter is supposed to distribute in the measurement space of volume $V$ uniformly. Therefore, the probability of observing $y_k$ conditioned on clutter is
\begin{eqnarray}\label{EqClutter}
\mathrm{P}(y_k|c_k=0) &=& 1/V
\end{eqnarray}
This paper employs RBPF for target detection and tracking in monitoring nodes and proposes a modified covariance intersection for data fusion in processing nodes.

\section{Rao-blackwellized particle filtering}
\label{section-RBPF}
Rao-blackwellized particle filtering (RBPF) is used for target tracking in a cluttered environment \cite{sarkka2007rao}. The RBPF divides tracking problem into data association problem and tracking single target problem. Data association and tracking are done by sequential importance resampling (SIR) and Kalman filter, respectively. This section briefly reviews the RBPF from \cite{sarkka2007rao,sarkka2004rao} where one or unknown number of targets are in the area of interest.

\subsection{One target scenario}
In this situation, suppose that only one target and clutter exist in the observed area. Consider a set of $N$ particles, where each particle $i$ at time step $k$ contains the following components:
\begin{eqnarray}
\{c^i_k, m^i_k, P^i_k, w^i_k\}
\end{eqnarray}
where $c^i_k,\ m^i_k,P^i_k$ and $w^i_k$ are data association indicator, mean and covariance of target states, and importance weight of the particle, respectively. The RBPF algorithm for one target with dynamic model as Eq.\ref{EqSystem} is:
\begin{enumerate}
\item
    The first step is to predict the states of target by the Kalman filter prediction step for each particle separately:
    \begin{eqnarray}
    \left[m^{-i}_k,P^{-i}_k\right]={\mathrm{KF}}_p\left(m^i_{k-1},P^i_{k-1},A_{k-1},Q_{k-1}\right)
    \end{eqnarray}
    where $m^{-i}_k$ and $P^{-i}_k$ denotes the prior mean and covariance, respectively. ${\mathrm{KF}}_p$ refers to the prediction step of Kalman filter, i.e. $m^{-i}_k=A_{k-1}m^i_{k-1}$ and $P^{-i}_k=A_{k-1}P^i_{k-1}A^T_{k-1}+Q_{k-1}$.
\item
    Now, consider a new measurement $y_k$ by the measurement model as Eq.\ref{EqMeasurement} is obtained. This step is to sample new association from the optimal importance distribution, that is to say sample a new association from the optimal importance distribution: $c^i_k\sim\mathrm{P}\left(c^i_k|y_{1:k},c^i_{1:k-1}\right)$, where the indices $1:k$ means from time step $1$ to $k$ . Assume that data association at time step $k$ is independent of the previous measurements $y_{1:k-1}$ and depends only on $m$ previous associations $c^i_{k-m:k-1}$.
    Thus, by using Bayes' rule, the optimal importance distribution is computed by calculating the likelihood of measurements for not only the target ($c^i_k=1$) but also the clutter ($c^i_k=0$) as follow:
    \begin{eqnarray}
    \mathrm{P}\left(c^i_k|y_{1:k},c^i_{1:k-1}\right) &\propto& \mathrm{P}\left(y_k|c^i_k,y_{1:k-1},c^i_{1:k-1}\right)\mathrm{P}\left(c^i_k|c^i_{k-m:k-1}\right)
    \end{eqnarray}
    where the likelihood of measurement is:
    \begin{eqnarray}
    \mathrm{P}\left(y_k|c^i_k,y_{1:k-1},c^i_{1:k-1}\right)  &=&
    \left\{
        \begin{array}{ll}
            \frac{1}{V} & \hbox{\rm{if} $c^i_k=0$} \\
            \rm{KF}_{lh}\left(y_k,m^{-i}_k, p^{-i}_k, H_k, R_k\right) & \hbox{\rm{if} $c^i_k=1$}
        \end{array}
    \right.
    \end{eqnarray}
    and ${\rm KF}_{lh}=P(y_k|y_{1:k-1})$ denotes the Kalman filter measurement likelihood.

    The new association $c^i_k=j$ is drawn with the probability $\pi^i_j$ as:
    \begin{eqnarray}
    \pi^i_j &=& \frac{\mathrm{P}(y_k|c^i_k=j,y_{1:k-1},c^i_{1:k-1})\mathrm{P}\left(c^i_k=j|c^i_{k-m:k-1}\right)}{\sum_{j'=0}^{1}\mathrm{P}(y_k|c^i_k=j',y_{1:k-1},c^i_{1:k-1})\mathrm{P}\left(c^i_k=j'|c^i_{k-m:k-1}\right)}
    \end{eqnarray}
\item
    The weight of each particle is updated by:
    \begin{eqnarray}
    w^i_k & = & w^{i}_{k-1} \frac{\mathrm{P}\left(y_k|c_k^i,y_{1:k-1},c^i_{1:k-1}\right)\mathrm{P}\left(c^i_k|c^i_{k-m:k-1}\right)}
    {\mathrm{P}\left(c^i_k|y_{1:k},c^i_{1:k-1}\right)}
    \end{eqnarray}
\item
    Finally, kalman filter update step is performed for each particle:
    \begin{eqnarray}
    \left[m^{i}_k,P^{i}_k\right]={\mathrm{KF}}_u\left(m^{-i}_{k},P^{-i}_{k},y_k,H_{k},R_{k}\right)
    \end{eqnarray}
    which means $m^i_k = m^{-i}_k+K^i_k V^i_k$ and $P^i_k = P^{-i}_k-K^i_kS^i_k{\left[K^i_k\right]}^T$ where $V^i_k=y_k-H_km^{-i}_k$, $S^i_k=H_kP^{-i}_k H^T_k+R_k$, and $K^i_k=P^{-i}_k H^T_kS^{-1}_k$.
    \end{enumerate}

It worth to note that, resampling is used to remove particles with very low weights and duplicate the particles with high weights. The resampling is applied based on the effective number of particles, which is estimated from the variance of the particle weights \cite{sarkka2007rao}.

\subsection{Unknown number of targets}
In this scenario, there are unknown number of targets with clutter in the observed area. Suppose that one target may born, some targets may die, or no birth-death happens in each time step.
The problem of tracking unknown number of targets by Rao-Blackwellized Monte Carlo Data Association (RBMCDA) method is divided into three sub problems: estimating the number of targets, data association, and tracking single target \cite{sarkka2007rao,sarkka2004rao}.

Suppose target dynamic, measurements model, and clutter distribution as Eqs.\ref{EqSystem}, \ref{EqMeasurement}, and \ref{EqClutter}, respectively. Possible events between two measurements $y_{k-1}$ and $y_k$ are a target birth, one or more targets death, and no target death. The measurement $y_k$ may be associated to clutter, one of the existing targets or a newborn target. Also, association priors are known and may be modeled as an $m$-th order Markov chain $\mathrm{P}(c_k|c_{k-m:k-1},T_{k-m:k-1})$ where $T_{k-m:k-1}$ contains number of targets at time steps $k-m$ to $k-1$. The RBMCDA algorithm is:
\begin{enumerate}
\item
    Suppose target birth happen with the probability $p_b$ when a new measurement is obtained and $b_k$ is the birth event indicator.
    So, There are three cases with different probabilities for data association and birth event as follow:
    \begin{itemize}
      \item Case 1: A target is born and the measurement is associated with the newborn target, that is to say $b_k=birth, c_k=T_{k-1}+1$
      \item Case 2: A target is not born and the measurement is associated with one of the existing targets or with clutter, that is to say $b_k=no\ birth$ and $c_k=j$ where $j\in 0,\cdots,T_k$. Markov model for data association in the case of no birth is $\mathrm{P}(c_k|no\ birth) = \mathrm{P}(c_k|c_{k-m:k-1})$.
      \item Case 3: Other events have zero probability
    \end{itemize}
    Thus, given the associations $c_{k-m:k-1}$ on the $m$ previous steps, the joint distributions of the event $b_k\in\{no\ birth,birth\}$ and the association $c_k$ is given by:
    \begin{eqnarray}
      \mathrm{P}(b_k,c_k|c_{k-m:k-1})&=&
       \left\{
                                         \begin{array}{ll}
                                           p_b & \hbox{case 1} \\
                                           (1-p_b)\mathrm{P}(c_k|c_{k-m:k-1}) & \hbox{case 2} \\
                                           0 & \hbox{case 3}
                                         \end{array}
                                       \right.
    \end{eqnarray}
\item
    The purpose of the death model is only to remove the targets with which no measurements have been associated for a long time. Death events are independent of measurements. After associating a measurement with a target, life time $t_d$ of the target has a probability density $\mathrm{P}(t_d)$ which usually is gamma distribution. If the last association with target $j$ was at the time $\tau_{k,j}$ and there is a sampled hypothesis that the target is alive on the previous time step $t_{k-1}$, then the  probability of target death at current time step $t_k$ is
    \begin{eqnarray} \label{death_model}
      \mathrm{P}(\mathrm{death\ of\ j}|t_k,t_{k-1},\tau_{k,j})&=&
      \mathrm{P}(t_d\in[t_{k-1}-\tau_{k,j},t_k-\tau_{k,j}]|t_d \geq t_{k-1}-\tau_{k,j})
    \end{eqnarray} 
\item
    Now, the RBMCDA with an unknown number of targets fits to the RBPF framework for a joint state $X_k$ which contains the states of the $T_k$ targets $X_k = \left[x_{k,1}, \cdots, x_{k,T_k}\right]^T$. The implementation idea is to assume that always a (very large) constant number of targets $T\infty$ exists. But an unknown, varying number of them are visible (or alive), and they are the ones we are tracking. The visibility indicator $e_k$ represents the visibility of targets. When a target birth occurs, a new item is set in $e_k$, and the corresponding target prior distribution is updated by the measurement. When a target dies, its distribution again becomes the prior and its state is moved to the end of $X_k$ and $e_k$.

    Here, the RBPF consists of a set of $N$ particles where each particle $i$ at time step $k$ contains the following components:
    \begin{eqnarray}
    \{c_{k-m+1:k}^i, e_k^i, m_{k,1:T_k}^i, P_{k,1:T_k}^i, w_k^i\}
    \end{eqnarray}
    where $c_{k-m+1:k}^i$, $e_k^i$, $m_{k,1:T_k}^i$, $P_{k,1:T_k}^i$ and $w_k^i$ are data association indicator, target visibility indicator, mean and covariance of the target, and importance weight of the $i$-th particle, respectively. The indices $1:k$ means for targets $1$ to $k$. The following information is also stored for each particle:
    \begin{eqnarray}
    \{T^i_k, \tau^i_{k,j}, \textrm{id}^i_{k,j}\}
    \end{eqnarray}
    where $T^i_k$, $\tau^i_{k,j}$, and $\textrm{id}^i_{k,j}$ are the number of targets, the time of the last measurement associated with target $j$, and a unique integer valued identifier in all particles, which is assigned at the birth of the target.
    The RBPF algorithm is now applied to each targets as:
    \begin{enumerate}
    \item
        The Kalman filter prediction step is applied on each target in each particle separately, due to independence of targets.
    \item
        Distribution $\mathrm{P}(e_k|e_{k-1})$ defines the dynamic of birth and deaths. Now, the data association model is of the form $\mathrm{P}(c_k|c_{k-m:k-1},e_k)$. So, the joint Markov chain model for the indicators is:
        \begin{eqnarray}
          \mathrm{P}(e_k,c_k|c_{k-m:k-1},e_{k-m:k-1}) &=& \mathrm{P}(c_k|c_{k-m:k-1},e_k )\mathrm{P}(e_k|e_{k-1})
        \end{eqnarray}
        The new association $c^i_k=j$ is drawn with the probability $\pi^i_j$ as:
    \begin{eqnarray}
    \pi^i_j &=& \frac{\mathrm{P}(y_k|e_k,c^i_k=j,y_{1:k-1},c^i_{1:k-1})\mathrm{P}\left(e_k,c^i_k=j|c^i_{k-m:k-1},e_{k-m:k-1}\right)}
         {\sum_{j'=0}^{T^i_k}\mathrm{P}(y_k|e_k,c^i_k=j',y_{1:k-1},c^i_{1:k-1})\mathrm{P}\left(e_k,c^i_k=j'|c^i_{k-m:k-1},e_{k-m:k-1}\right)}\nonumber
    \end{eqnarray} 
        
    \item
        The weight of each particle is updated similar to RBPF.
    \item
        The measurement update is also performed for each target separately.
    \end{enumerate}
\end{enumerate}
The assumptions which considered in derivation of RBMCDA are not very restrictive in practice. For example, the assumption that targets are independent from each other is a reasonable assumption in reality. The assumption that possible events between two measurements $y_{k-1}$ and $y_k$ are a target birth, one or more targets death, and no target death restricts our problem to only one target possibly born in each time step which is not very restrictive in practice. These assumptions are quite much the same as typical ones in tracking literature \cite{bar1995multitarget,sarkka2004rao}.

\section{Modified-CI for data fusion}
\label{section-modCI}
Traditional CI is a convex combination of mean and covariance of estimates received from monitoring nodes and gives more accurate estimations. Suppose that there are two pieces of information $A$ and $B$ which are received from different sources. The only available information is estimation of the means, covariance and the cross-correlations between them which are $\{m^a,P^a\}$ and $\{m^b,P^b\}$ and $P^{ab}$. If the estimations are independent (that is to say $P^{ab}=0$), then the conventional Kalman is an optimal filter. If $P^{ab}\ne 0$ is known then the Kalman filter with colored noise may be the best option but if $P^{ab}$ is unknown then CI is a consistent choice. CI fuses two pieces of information into $\{m^c,P^c\}$ as \cite{chen2002estimation,liggins2008handbook}:
\begin{eqnarray}\label{EqCI}
m^c &=& P^c\left(\omega (P^a)^{-1}m^a+(1-\omega)(P^b)^{-1}m^b\right) \nonumber\\
(P^c)^{-1} &=& \omega (P^a)^{-1}+(1-\omega)(P^b)^{-1}
\end{eqnarray}
where $\omega$ is computed to minimize a selected norm such as trace or determinant of $P^c$. While this minimization requires numeric solution of a convex optimization problem (see "fmincon" in Matlab optimization toolbox) for higher dimensions, closed-form solution is given for lower dimensions in \cite{reinhardt2012closed}. It is proved in \cite{hall2001multisensor} that CI yields a consistent estimate for any value of $P^{ab}$, i.e. $P^c-\bar{P}^c\geq0$ where $\bar{P}^c$ is true covariance matrix. The only constraint that is imposed on the assumed estimates is consistency. In other words, $P^a-\bar{P}^a\geq0$ and $P^b-\bar{P}^b\geq0$ where $\bar{P}^a$ and $\bar{P}^b$ are true covariance matrices.

Now, back to the problem stated in section \ref{section-ProblemStatement}. In a distributed network, RBPF is implemented in each monitoring node to detect and estimate the states of possible targets using data measured by its own sensors. Then, the estimations are fused together in processing nodes, see Fig.\ref{DataFusion}. Some researches aim to reach consensus over the network \cite{hlinka2014distributed}. But, this paper concerns on taking detection parameters into account in fusion algorithm. 
%So, data fusion is considered for two nodes and each node fuses received data by its own estimates and is not worried about redundant information due to employing CI.
\begin{figure}[htb]
	\begin{center}
		\includegraphics[width=0.9\textwidth]{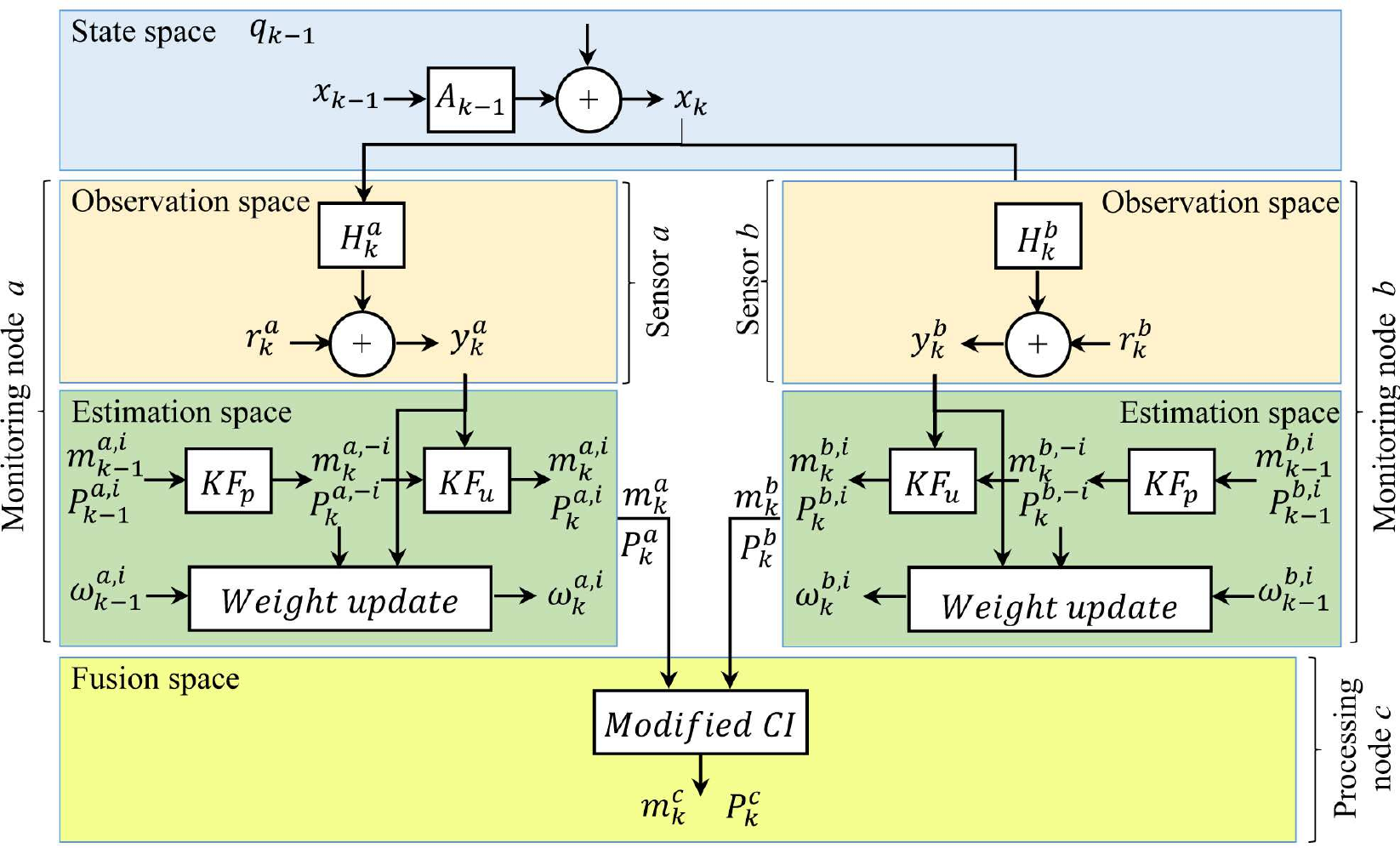}
		\caption{Fusion space}
		\label{DataFusion}
	\end{center}
\end{figure}

The following theorem proposes a modified-CI data fusion method for two nodes $a$ and $b$ considering different probability of detection for each node. It modifies the traditional CI by replacing the covariance matrices $P^a$ and $P^b$ with $\widehat{P}^a$ and $\widehat{P}^b$. The modified covariance matrices $\widehat{P}^a$ and $\widehat{P}^b$ depend on covariance matrices of estimates ($P^a$ and $P^b$), detection probability of both nodes, and parameters $\chi_a$ and $\chi_b$ which are related to the growth of covariance matrices in miss-detection time intervals. If there is no detection in time, then the accuracy of estimates decreases. So, the covariance of estimations are expanded with the growth parameter $\chi^{-1}\geq1$.

\begin{theorem}
    Let, two monitoring nodes $a$ and $b$ with probability of detection $p$ and $q$ estimate mean and covariance $\{m^a,P^a\}$ and $\{m^b,P^b\}$ for any observed target, respectively. Then, the following data fusion is consistent for any $P^{ab}$:
    \begin{eqnarray}
    \widehat{m}^c &=& \widehat{P}^c\left(\omega(\widehat{P}^a)^{-1}m^a+(1-\omega)(\widehat{P}^b)^{-1}m^b\right)\nonumber \\
    (\widehat{P}^c)^{-1} &=&\omega(\widehat{P}^a)^{-1}+(1-\omega)(\widehat{P}^b)^{-1}
    \end{eqnarray}
    where $(\widehat{P}^a)^{-1} = \{p+(1-p)\chi_a(q+(1-q)\chi_b)\}(P^a)^{-1}$, $(\widehat{P}^b)^{-1} = \{q+(1-q)\chi_b(p+(1-p)\chi_a)\}(P^b)^{-1}$, and $\chi_a^{-1}$, $\chi_b^{-1}$ express the growth of covariance in miss-detection intervals. $\omega$ is selected to minimize the trace of $\widehat{P}^c$.
\end{theorem}
\begin{proof}
    We articulate a mixed game between two nodes $a$ and $b$ as players. Each player has two action: detect or miss the targets. Player $1$ and $2$ detect with probabilities $p$ and $q$ and miss with $(1-p)$ and $(1-q)$, respectively, see table \ref{game}. Since, the aim of target tracking is to minimize the trace of covariance matrix, So the trace of $(P^c)^{-1}$ is considered as payoff function. The idea behind mixed games is that the solution gives a probability distributions which players select their actions randomly according to that \cite{hespanha2011introductory}. But, here the probabilities of selecting an action (i.e. probability of detection) is given. So, we have to manipulate the payoff function such that the probabilities $p$ and $q$ result to maximum payoff. Table \ref{game} shows payoff matrix of a 2-player mixed game which the payoff of player $1$ in only written.

    \begin {table}[htb]
    \caption{Payoff matrix of two-node data fusion game} \label{game}
    \begin{center}
    \begin{tabular}{|c|c|c|c|}
      \hline
      \multicolumn{2}{|c|}{Data fusion}& \multicolumn{2}{|c|}{Node b (Player 2)} \\ \cline{3-4}
      \multicolumn{2}{|c|}{Game}& $q$  & $1-q$ \\
      \hline
      Node a     & $p$   & $tr\{\omega(P^a)^{-1}+(1-\omega)(P^b)^{-1}\}$  &  
      $tr\{\omega(P^a)^{-1}+(1-\omega)\chi_b(P^b)^{-1}\}$ \\ \cline{2-4}
      (Player 1) & $1-p$ & $tr\{\omega\chi_a(P^a)^{-1}+(1-\omega)(P^b)^{-1}\}$  & $tr\{\chi_a\chi_b(\omega(P^a)^{-1}+(1-\omega)(P^b)^{-1})\}$  \\
      \hline
    \end{tabular}
    \end{center}
    \end{table}
    
    Anyway, the mean of player $1$ utility could be written as:
    \begin{eqnarray}
        E\{\pi_1\} &=& tr\{pq\{\omega(P^a)^{-1}+(1-\omega)(P^b)^{-1}\}+\\\nonumber
        &&p(1-q)\{\omega(P^a)^{-1}+(1-\omega)\chi_b(P^b)^{-1}\}+\\ \nonumber
        &&(1-p)q\{\omega\chi_a(P^a)^{-1}+(1-\omega)(P^b)^{-1}\}+\\\nonumber
        &&(1-p)(1-q)\chi_a\chi_b\{\omega(P^a)^{-1}+(1-\omega)(P^b)^{-1}\}\}\\ \nonumber
        &=&tr\{\omega\{p+(1-p)\chi_a(q+(1-q)\chi_b)\}(P^a)^{-1}+\\\nonumber
        &&(1-\omega)\{q+(1-q)\chi_b(p+(1-p)\chi_a)\}(P^b)^{-1}\}\\ \nonumber
    \end{eqnarray}
    Matching this equation with Eq.\ref{EqCI} leads to define $(\widehat{P}^c)^{-1}=\omega(\widehat{P}^a)^{-1}+(1-\omega)(\widehat{P}^b)^{-1}$ where $(\widehat{P}^a)^{-1} = \{p+(1-p)\chi_a(q+(1-q)\chi_b)\}(P^a)^{-1}$ and $(\widehat{P}^b)^{-1} = \{q+(1-q)\chi_b(p+(1-p)\chi_a)\}(P^b)^{-1}$. 
    So, we have
    \begin{eqnarray}
        E\{\pi_1\} &=& tr\{\widehat{P}^c)^{-1}\}\\ \nonumber
    \end{eqnarray}
    Finally, finding $\omega$ in such a way that maximize the trace of $(\widehat{P}^c)^{-1}$ leads to maximum utility.
    
    if $P^a-\bar{P}^a\geq0$ and $P^b-\bar{P}^b\geq0$ then $\widehat{P}^a-\bar{P}^a\geq0$ and $\widehat{P}^b-\bar{P}^b\geq0$ because $\{p+(1-p)\chi_a(q+(1-q)\chi_b)\}\leq1$ and $\{q+(1-q)\chi_b(p+(1-p)\chi_a)\}\leq1$ which result in $\widehat{P}^a\geq P^a$ and $\widehat{P}^b\geq P^b$. Now, by considering these assumptions, the consistency could be proved same as \cite[Appendix 12.A]{hall2001multisensor}.
\end{proof}
Let us investigate two particular case. First, $\chi_a=\chi_b=1$ which means that the covariance matrix does not growth on miss-detections, that is to say, there is no difference between detection or miss-detection as players possible actions. So, $(\widehat{P}^a)^{-1}=(P^a)^{-1}$ and $(\widehat{P}^b)^{-1}=(P^b)^{-1}$ which leads to the traditional CI. This make sense because we make no difference between two different action of players. Second, if $\chi_a=\chi_b$ and $p=q$ which means that both nodes have identical probability of detection, then $(\widehat{P}^a)^{-1}=\alpha (P^a)^{-1}$ and $(\widehat{P}^b)^{-1}=\alpha (P^b)^{-1}$ where $\alpha=\{p+p(1-p)\chi_a+(1-p)^2\chi_a^2)\}$. It is reasonable that under similar parameters for nodes, the covariance matrix $(P^a)^{-1}$ weights similar to $(P^b)^{-1}$.

Batch Covariance Intersection (BCI) is an extension of the pairwise CI for fusion of more than two data \cite{julier2009general}. When the mean $m^i$ and variances $P^i$ $(i=1,\ldots,N)$ are known but cross-covariances $P^{ij}$ $(i,j=1,\ldots,N; i\neq j)$ are unknown, the BCI is $m^{BCI}=P^{BCI}\sum_{i=1}^{N}\omega_i (P^i)^{-1}m^i$ and $(P^{BCI})^{-1}=\sum_{i=1}^{N}\omega_i (P^i)^{-1}$ with constraints $0\leq\omega_i\leq1$ and $\sum_{i=1}^N\omega_i=1$. In the problem stated in this paper, BCI is modified as following lemma:
\begin{lemma}
	Let, monitoring nodes $i$ $(i=1,\ldots, N)$, with probability of detection $p_i$ estimate mean and covariance $\{m^i,P^i\}$ for any observed target. The following data fusion (which we call Modified BCI: MBCI) is consistent for any cross covariance $P^{ij}$ $(i,j=1,\ldots,N; i\neq j)$:
	\begin{eqnarray}
	\widehat{m}^{MBCI} &=& \widehat{P}^{MBCI}\sum_{i=1}^{N}\omega_i(\widehat{P}^i)^{-1}m^i\nonumber \\
	(\widehat{P}^{MBCI})^{-1} &=&\sum_{i=1}^{N}\omega_i(\widehat{P}^i)^{-1}
	\end{eqnarray}
	where $(\widehat{P}^i)^{-1} = \{p_i+(1-p_i)\chi_i\prod_{j=1,j\neq i}^{N}(p_j+(1-p_j)\chi_j)\}(P^i)^{-1}$ and $\chi_i^{-1}$ express the growth of covariance in miss-detection intervals. $\omega$ is selected to minimize the trace of $\widehat{P}^i$.
\end{lemma}
\begin{proof}
	The proof is simply given by extend the game which is defined in theorem 4.1 for $N$ players.
\end{proof}

\section{simulation  result}
\label{section-Simulation}
In this section, the proposed method is utilized when there is one target or unknown number of target in the are of interest.

\subsection{One target}
In this case, there is one target and clutter in the area of interest. Consider the target dynamic as in \cite{sarkka2007rao}:
\begin{eqnarray}\label{ExSystem}
\dot{x} &=& \left[
                \begin{array}{cccc}
                0 & 0 & 1  & 0 \\
                0 & 0 & 0  & 1 \\
                0 & 0 & 0  & a \\
                0 & 0 & -a & 0
                \end{array}
            \right]x+\left[
                \begin{array}{cc}
                0 & 0 \\
                0 & 0 \\
                1 & 0 \\
                0 & 1
                \end{array}
            \right]w.
\end{eqnarray}
where $a$ is a constant parameters in each time step but may change in time; $w$ is zero mean Gaussian distribution with covariance matrix $Q=0.1I$ where $I$ is a $2\times2$ identity matrix. Assume that two monitoring node $a$ and $b$, with different probability of detection $p=0.9$ and $q=0.7$, observe the area of interest by measurement model:
\begin{eqnarray}
y       &=& \left[
                \begin{array}{cccc}
                1 & 0 & 0 & 0 \\
                0 & 1 & 0 & 0
                \end{array}
            \right]x+\left[
                \begin{array}{cc}
                1 & 0 \\
                0 & 1
                \end{array}
            \right]v
\end{eqnarray}
where $v$ is zero mean Gaussian distribution with covariance matrix $R=0.05I$.
Density of clutter is also considered as $0.5$ for both nodes. Monitoring nodes sample data every $dt=0.025$ second. $\chi_a^{-1}$ and $\chi_b^{-1}$ are considered to be $1+dt$. An instantiation of the target trajectory and the corresponding measurements is shown in Fig.\ref{TargetPath}.

\begin{figure}[htb]
\begin{center}
\includegraphics[width=0.5\textwidth]{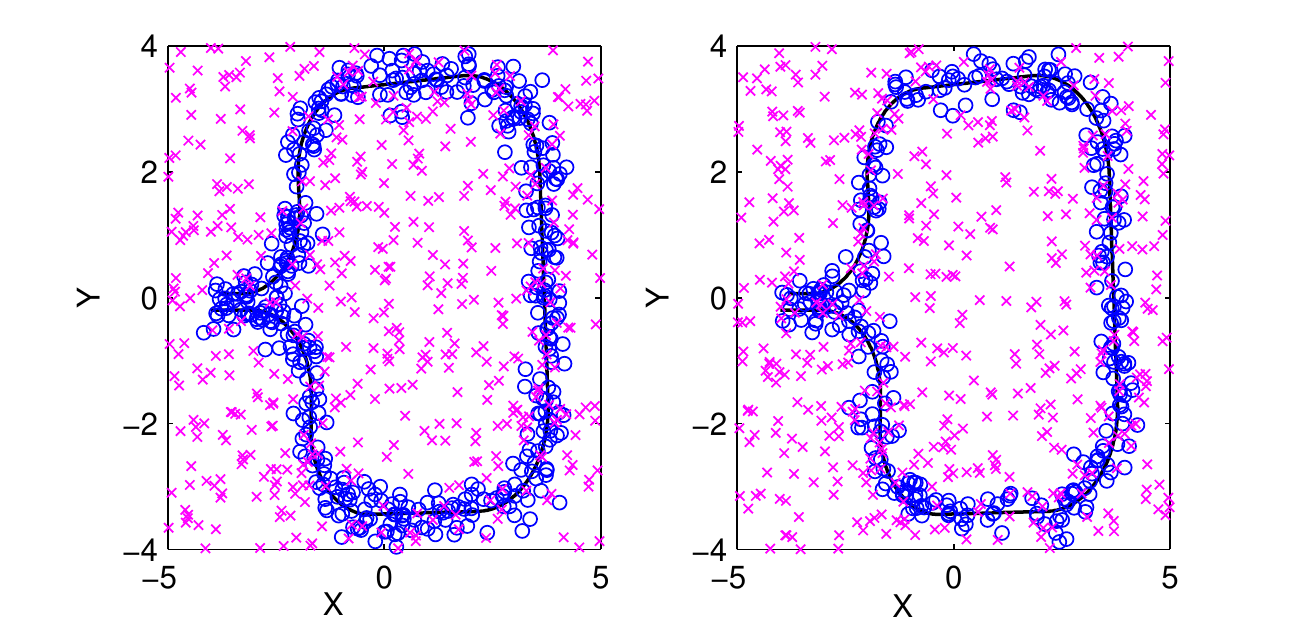}
\caption{An instantiation of the target path and measurements. (Left) node $a$; (Right) node $b$. Solid line (Black): target path; o(Blue): target measurements;$\times$(Pink): clutter measurements.}
\label{TargetPath}
\end{center}
\end{figure}

Monitoring nodes $a$ and $b$ employe RBPF with $20$ particles to estimate states of the target which are shown in Fig.\ref{ParticleFiltering}.
Then, the estimations are transmitted to processing node $c$ where fuses the recieved data by modified CI. The results are shown in Fig.\ref{FusedModCI}.

\begin{figure}[htb]
\begin{center}
\includegraphics[width=0.5\textwidth]{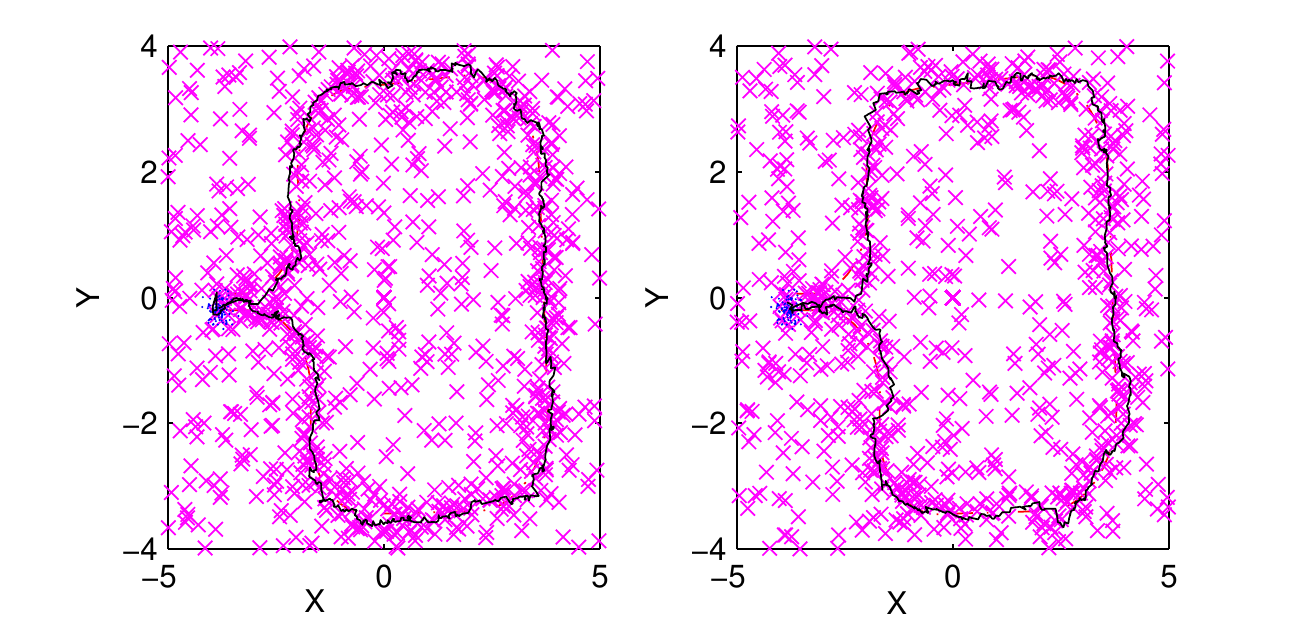}
\caption{Estimation of states by RBPF. (Left) Node a; (Right) Node b. Dashed line (red): target path; Solid line(Black): estimated path; $\times$ (Pink): all measurements.}
\label{ParticleFiltering}
\end{center}
\end{figure}
\begin{figure}[htb]
\begin{center}
\includegraphics[width=0.3\textwidth]{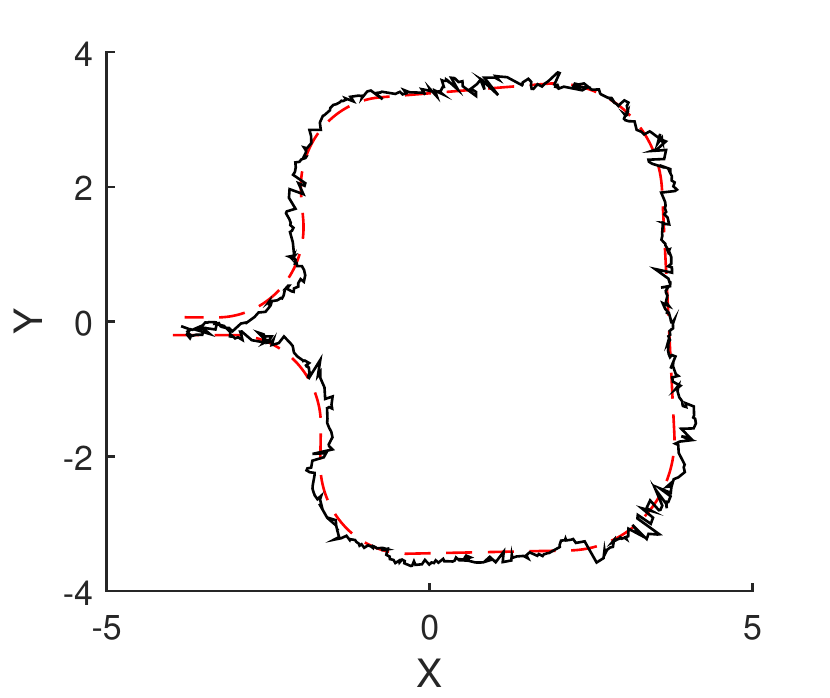}
\caption{Data fusion by modified-CI . Dashed line (red): target path; Solid line(Black): estimated path by data fusion.}
\label{FusedModCI}
\end{center}
\end{figure}

The mean square error (MSE) of estimations as well as the mean of covariance matrix (MNCM) in the terms of 2-norm are calculated for using Kalman filter and RBPF in monitoring nodes for state estimations and CI and modified-CI in processing nodes for data fusion to show the utility of the proposed method. Results are given in table.\ref{table-Comparision}. There is no data association in Kalman filter; so the MSE of Kalman filter is large due to existence of clutter. By considering data association in RBPF, the MSE improved significantly rather than Kalman filter. MNCM is decreased after fusing the estimations of two nodes by CI. Finally, MSE and MNCM for modified-CI are less than other methods. The 2-norm of covariance matrix is shown in Fig.\ref{CovarianceMatrix} for time period between 13 to 15 (sec). It can be seen that the 2-norm of covariance matrix for modified-CI is lower than other method.
\begin{figure}[htb]
\begin{center}
\includegraphics[width=0.38\textwidth]{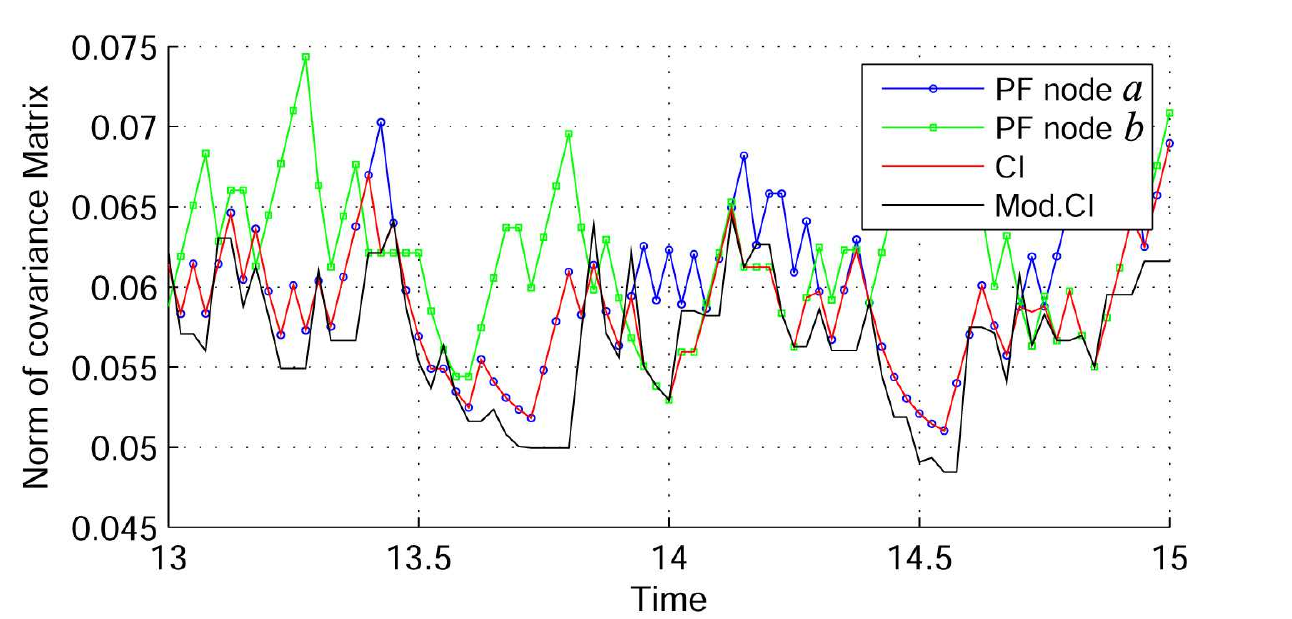}
\caption{2-norm of covariance matrix}
\label{CovarianceMatrix}
\end{center}
\end{figure}

\begin {table}[htb]
\caption {MSE and MNCM for different algorithm in the situation of one target } \label{table-Comparision}
\begin{center}
\begin{tabular}
{|l|l|l|} \hline
\textbf{Filtering Method} & \textbf{MSE} & \textbf{MNCM} \\ \hline
Kalman (node $a$) & 1.447 & - \\ \hline
Kalman (node $b$) & 1.613 & - \\ \hline
RBPF (node $a$) & 0.141 & 0.0632 \\ \hline
RBPF (node $b$) & 0.161 & 0.0666 \\ \hline
CI & 0.141 & 0.0613 \\ \hline
Modified CI & 0.109 & 0.0592 \\ \hline
\end{tabular}
\end{center}
\end {table}

\subsection{Unknown number of targets}
Assume unknown number of targets with dynamic Eq.\ref{ExSystem} exist, born or die in an area of interest which is observed by distributed non-homogeneous monitoring network of Fig.\ref{Network}. The scenario that is simulated here consider one target in the interested area and another $3$ targets appear at $t=1$ and disappear one by one at $t=3,3.5,4$. The network consist of $7$ monitoring nodes which utilize RBMCDA to estimate number of targets and their states, and $3$ processing nodes which fuse received data by modified CI. Each monitoring node send its estimation to the connected neighbor processing nodes in every time step. The detection probability of monitoring nodes are $[0.8, 0.8, 0.9, 0.75, 0.95, 0.9, 0.85]$ which are used in modified CI. Any monitoring nodes may detect some targets and/or clutter in any time step. RBMCDA assume that a target is born in each time step with probability $p_b=0.01$ and it is removed based on the death model of Eq.\ref{death_model} where $t_d$ has gamma distribution with $\alpha=2$ and $\beta=0.5$. An instantiation of this scenario with considering $20$ particles in the filter for a typical monitoring node (node $3$) is shown in Fig.\ref{FigMultiTargetEstimation}. The plot on left shows true target path and all the measurements while the middle plot shows estimated path. The right plot depicts the number of detected targets in time.
\begin{figure}[htb]
\begin{center}
\includegraphics[width=0.7\textwidth]{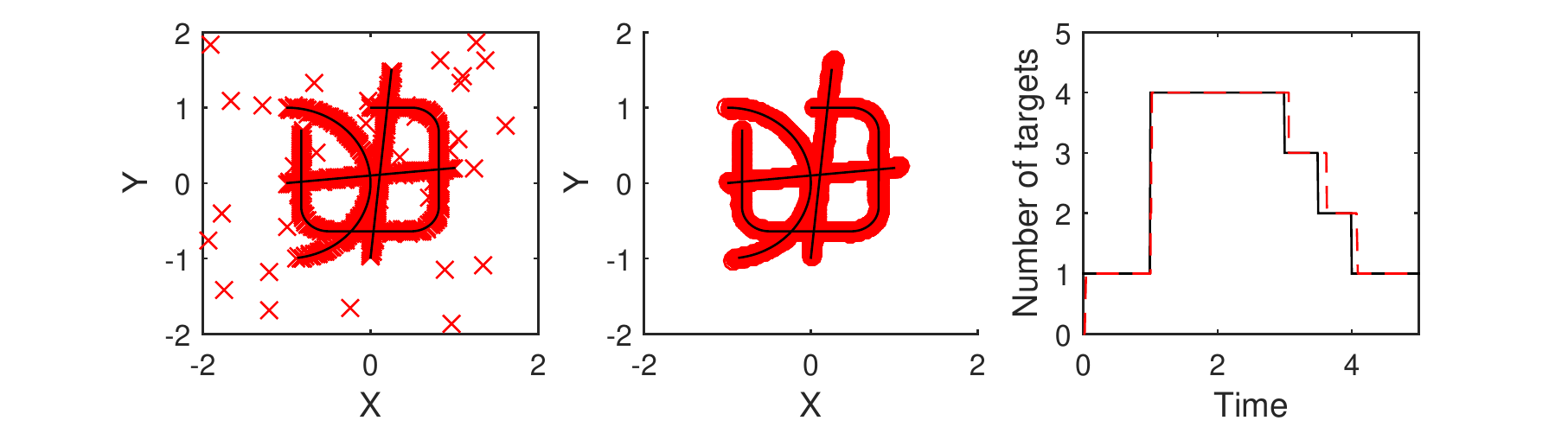}
\caption{(Left) solid lines (black): targets path; $\times$(red): all measurements; (Middle) solid lines (black): targets path; o(red) estimated path (Right) solid line (black): Number of targets, dashed lines (red): estimated number of targets.}
\label{FigMultiTargetEstimation}
\end{center}
\end{figure}

The processing nodes have to associate received data together which is done here based on minimum distance, namely data which have minimum distance from each other and are close together more than a threshold. Modified-CI is conducted for data fusion if more than one data is received for any target. The output results are shown in Fig.\ref{FigMultiTargetModCI} for processing node number 9. The sum of MSE and MNCM for all detected targets are calculated for different nodes which are written in table \ref{table-Comparision-2}. The proposed method has less MSE and MNCM than other algorithms.
\begin{figure}[htb]
\begin{center}
\includegraphics[width=0.3\textwidth]{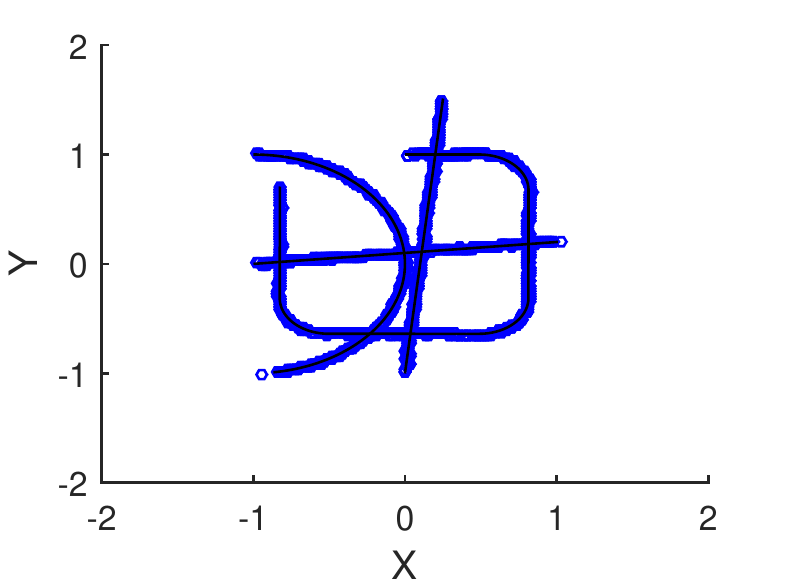}
\caption{Modified CI. solid lines(black): targets path; o(blue): fused estimations.}
\label{FigMultiTargetModCI}
\end{center}
\end{figure}

\begin {table}[htb]
\caption {MSE and MNCM for different algorithm in the situation of unknown number of targets } \label{table-Comparision-2}
\begin{center}
\begin{tabular}
{|l|l|l|l|} \hline
\textbf{Node number} & \textbf{Algorithm} & \textbf{MSE ($\times10^{-3}$)} & \textbf{MNCM} \\ \hline
1 & RBPF & 0.1430 & 0.0578 \\ \hline
2 & RBPF & 0.4716 & 0.0616 \\ \hline
3 & RBPF & 0.1228 & 0.0533 \\ \hline
4 & RBPF & 0.1830 & 0.0580 \\ \hline
5 & RBPF & 0.1010 & 0.0511 \\ \hline
6 & RBPF & 0.1166 & 0.0529 \\ \hline
7 & RBPF & 0.1912 & 0.0556 \\ \hline
8 & CI   & 0.0821 & 0.0508 \\ \hline
9 & CI   & 0.0731 & 0.0560 \\ \hline
10& CI   & 0.0884 & 0.0505 \\ \hline
8 & Modified CI & 0.0784 & 0.0501 \\ \hline
9 & Modified CI & 0.0763 & 0.0523 \\ \hline
10& Modified CI & 0.0799 & 0.0493 \\ \hline
\end{tabular}
\end{center}
\end {table}

\section{ conclusion}
\label{section-Conclusion}
There are nodes with different probability of detection and clutter density in a non-homogeneous monitoring network which may results in different precision for state estimations. The fact of non-homogeneity monitoring nodes is considered in data fusion algorithm to reduce estimation error. To this end, a data fusion mixed game is articulated between two neighbor monitoring nodes which results to a modified CI. The proposed modified CI consider probability of detection of each node into account in fusion process. To show the utility of the method, a scenario for target detection and tracking in a cluttered environment was proposed where multi monitoring nodes with different probability of detection monitor an area of interest. The problem was investigated in two situation: one target and unknown number of targets. The Rao-Blackwellized particle filter was employed to estimate the states of targets in each monitoring node. The particle filter estimations from different monitoring nodes were fused in processing nodes by modified CI. The simulation results showed that mean square error and the 2-norm of covariance matrix for the proposed method are less than Kalman filter, particle filter and fused data with traditional CI.

\bibliographystyle{wileyj}
\bibliography{Citations}
\end{document}